\documentclass[twocolumn,noshowpacs,noshowkeys,pra,longbibliography,nofootinbib]{revtex4-1}%
\pdfoutput=1
\usepackage[caption=false]{subfig}
\usepackage{algorithm}
\usepackage{color}
\usepackage{algorithmic}
\usepackage{amsmath}
\usepackage{amssymb}
\usepackage{amsthm}
\usepackage{bm}
\usepackage{thmtools}
\declaretheorem{theorem}
\usepackage{array}
\usepackage{url}
\usepackage{hyperref}
\usepackage{graphicx}
\usepackage{bibentry}
\usepackage{microtype}

\providecommand{\U}[1]{\protect\rule{.1in}{.1in}}
\newtheorem*{theorem*}{Theorem}

\DeclareMathOperator{\trace}{Tr}

\DeclareMathOperator\arcsinh{arcsinh}

\newcommand{\ket}[1]{\left| #1 \right>} 
\newcommand{\bra}[1]{\left< #1 \right|} 

\begin{document}
\title{Fundamental limits of quantum-secure covert optical sensing}
\thanks{This research was funded by DARPA under contract number HR0011-16-C-0111, UK EPSRC (EP/K04057X/2) and the National Quantum Technologies Programme (EP/M01326X/1, EP/M013243/1).
This document does not contain technology or technical data controlled under either the U.S.~International Traffic in Arms Regulations or the U.S.~Export Administration Regulations.}
\author{Boulat A. Bash,$^{1}$ Christos N. Gagatsos,$^2$ Animesh Datta,$^2$ and Saikat Guha$^1$}
\affiliation{$^1$\textit{Quantum Information Processing Group, Raytheon BBN Technologies, Cambridge, Massachusetts, USA 02138},\\
$^2$\textit{Department of Physics, University of Warwick, Coventry CV4 7AL, United Kingdom}
}

\begin{abstract}
We present a square root law for active sensing of phase $\theta$ of a single 
  pixel using optical probes that pass through a single-mode lossy 
  thermal-noise bosonic channel.
Specifically, we show that, when the sensor uses an $n$-mode covert optical
  probe, the mean squared error (MSE) of the resulting estimator
  $\hat{\theta}_n$ scales as 
  $\langle (\theta-\hat{\theta}_n)^2\rangle=\mathcal{O}(1/\sqrt{n})$; improving
  the scaling necessarily leads to detection by the adversary with high
  probability.
We fully characterize this limit and show that it is achievable using laser 
  light illumination and a heterodyne 
  receiver, even when the adversary captures every photon that does not
  return to the sensor and performs arbitrarily complex measurement
  as permitted by the laws of quantum mechanics.
\end{abstract}
\maketitle

\section{Introduction}
\label{sec:intro}

Active probing with electromagnetic radiation is used in many practical systems to measure physical properties of objects.
However, there are scenarios where the detection of such probing by an unauthorized third party (which could be the target object) is undesired.
In these scenarios covert, or low probability of intercept/detection (LPI/LPD), signaling must be used.
While covertness is often required by practical stand-off sensing systems, the fundamental limits of sensing under the covertness constraints has been relatively under-explored.

Recently, the fundamental limits of covert communication have been characterized for several classical and quantum channels. Covert communication is governed by the \emph{square root law} (SRL): $\mathcal{O}(\sqrt{n})$ bits can be reliably transmitted in $n$ channel uses without being detected by the adversary; transmission of more bits results in either detection or uncorrectable decoding errors.  The SRL was first proven for the classical wireless channels subject to the additive white Gaussian noise (AWGN) \cite{bash13squarerootjsacisit}, with follow-on works extending this result to discrete memoryless channels (DMCs) and fully characterizing the constant hidden by the Big-$\mathcal{O}$ notation \cite{che13sqrtlawbscisit, bloch15covert, wang15covert}. 

Now consider the lossy thermal-noise bosonic channel, which is the quantum-mechanical model for optical communication. The SRL also governs covert communication over this channel: provided that there exists a noise source that the adversary does not control (for example, the unavoidable thermal noise from blackbody radiation at the operating temperature and wavelength), $\mathcal{O}(\sqrt{n})$ covert bits can be reliably transmitted using $n$ orthogonal spatio-temporal polarization modes. As in the SRL for AWGN channel, transmission of more bits results in either detection or uncorrectable decoding errors \cite{bash15covertbosoniccomm}.  Remarkably, the SRL is achievable using standard optical communication components (laser light modulation and homodyne receiver) even when the adversary has access to all the photons that are not captured by the legitimate receiver, as well as arbitrary quantum measurement, storage and computing capabilities.  Conversely, entangled photon transmissions, as well as arbitrary quantum measurement, storage and computing capabilities do not permit one to reliably transmit more covert bits than the SRL allows, even when the adversary has access to only a fraction of the transmitted photons and is only equipped with a noisy photon counting receiver.  

A covert communications adversary has to decide whether or not a transmission takes place.  Thus, the transmitter has to render the adversary's detector ineffective by ensuring that it can only do a little better than a random decision. The SRL for covert communications arises because, for this to happen, the average symbol power $\bar{n}_{\rm S}$ must scale in the blocklength $n$ as $\bar{n}_{\rm S}=\mathcal{O}(1/\sqrt{n})$.  In the AWGN setting, $\bar{n}_{\rm S}$ is the average squared symbol magnitude, while in the bosonic channel setting it is the mean photon number per mode.  By standard arguments, the total number of reliably transmissible bits thus scales as $n\bar{n}_{\rm S}=\mathcal{O}(\sqrt{n})$.  Since $\lim_{n\to\infty}\mathcal{O}(\sqrt{n})/n=0$, the covert communication channel capacity is zero, however, a non-trivial number of bits can be transmitted when $n$ is large (see a tutorial survey in \cite{bash15covertcommmag}).

The results for the fundamental limits of covert communication over lossy thermal-noise bosonic channel in \cite{bash15covertbosoniccomm} motivate our investigation of the fundamental limits of covert sensing.  We begin by noting that the most effective method of staying covert is passive imaging, which emits no energy. Passive imaging collects the scattered light from a naturally-illuminated (or self-luminous) scene. However, this can be impractical, or even impossible, in many scenarios. For example, the scene could be hidden from direct line of sight or the signal to noise ratio (SNR) at the receiver could otherwise be insufficient to obtain the desired performance. In these situations, active transmitters must be employed to illuminate the target.

\begin{figure}
\centering
\includegraphics[width=\columnwidth]{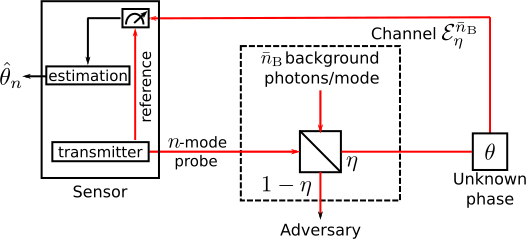}
\caption{Active probing of an unknown phase of a pixel.
Transmitted $n$-mode probe is corrupted by a lossy thermal-noise bosonic 
  channel with transmissivity $\eta$ and thermal background mean photon 
  number $\bar{n}_B$ per mode.
Fraction $1-\eta$ of the photons is lost and can be captured by the adversary,
  while the remaining fraction $\eta$ of the photons is received by the sensor
  after the probe acquires unknown phase $\theta$ in each mode.
An estimate $\hat{\theta}_n$ is computed from the measurement of the received
  probe state and the reference state (which adversary cannot access).
The input-output relationship of the bosonic channel is captured by 
  a beamsplitter of transmissivity $\eta$, with the sensor's transmitter at one
  of the input ports and the phase rotation followed by the sensor's receiver 
  at one of the output ports.
The other input and output ports of the beamsplitter correspond to the 
  environment and the adversary.
Switching the order of the phase rotation and the bosonic channel 
  does not affect phase estimation \cite[App.~A]{tsang13metrology}.
\label{fig:setup}}
\end{figure}

We therefore study the fundamental limits of quantum-secure covert active sensing.  This notion of security is more stringent than, for example, ensuring that the return probes are not spoofed by the target as done in \cite{malik12quantumsecuredimaging} (undetectable probes cannot be spoofed).  As illustrated in Figure \ref{fig:setup}, we explore covert estimation of an unknown phase $\theta$ of a single pixel using an optical probe that passes through a lossy thermal-noise bosonic channel with transmissivity $\eta$ and thermal background mean photon number $\bar{n}_B$ per mode.  Adversary captures up to $1-\eta$ fraction of light from the probe.  We assume that the distance to the target pixel is known. The focus on estimating the unknown phase allows us to leverage the extensive literature in quantum metrology (see \cite{demkowicz15quantmetrologysurvey} for a recent survey); however, we believe that similar results hold in other sensing modalities (such as ranging, reflectometry, target detection, and target classification). Ensuring covertness of transmitted probes imposes the same power constraint $\bar{n}_{\rm S}=\mathcal{O}(1/\sqrt{n})$ photons/mode as in communications. We thus find that covert sensing is subject to its own SRL:

\begin{theorem*}[Square-root law for covert phase sensing]
Suppose the sensor attempts to estimate an unknown phase $\theta$ of a pixel 
  using an $n$-mode optical probe that passes through a lossy thermal-noise
  bosonic channel, as described in Figure \ref{fig:setup}.
Also suppose that the adversary has
  access to fraction $1-\eta$ of the transmitted photons.
Then the sensor can achieve mean squared error (MSE)
  $\langle(\theta-\hat{\theta}_n)^2\rangle=\mathcal{O}(1/\sqrt{n})$ while
  ensuring the ineffectiveness of the adversary's detector.
Attempting to decrease scaling for MSE results in detection of
  the interrogation attempt with high probability.
\end{theorem*}

In addition to the scaling law above, we characterize the constants hidden by the Big-$\mathcal{O}$ notation for several covert estimation schemes.  We find that using laser pulse modulation and heterodyne receiver yields MSE that is at most twice that of the laser light modulation coupled with the optimal receiver, and a factor $\frac{2}{1-\eta}$ greater than the ultimate lower bound. This limit on enhancing the design coupled with the constraint on the power per mode imposed by the covertness requirement implies that only increasing the number of available orthogonal modes $n$ can improve the performance of covert sensing systems. 

After introducing the channel model and the background on our performance metrics in the next section, we prove the square root law for covert sensing of phase in Section \ref{sec:proof}.  We then conclude with a discussion of future work in Section \ref{sec:conclusion}.

\section{Prerequisites}\label{sec:prerequisites}

\subsection{Estimation}
Consider a single-mode lossy bosonic channel $\mathcal{E}_{\eta}^{\bar{n}_{\rm B}}$ with path transmissivity $\eta\in(0,1)$ and thermal noise mean photon number $\bar{n}_{\rm B}>0$, as depicted in Figure \ref{fig:setup}.
The sensor (an optical interferometer) interrogates the target using an $n$-mode probe with average photon number $\bar{n}_{\rm S}$ per mode, where $1-\eta$ fraction of these photons is lost to the adversary, while the remaining fraction $\eta$ returns to the sensor after acquiring the unknown phase $\theta$ on each mode.
The sensor estimates $\theta$ using the collected light and retained state (e.g., a local oscillator for a coherent detector), and outputs estimate $\hat{\theta}_n$.
The sensor has to minimize the MSE of the estimate $\langle (\theta-\hat{\theta}_n)^2\rangle$ while preventing the detection of the probe by the adversary.
The quantum Cramer-Rao lower bound (QCRLB) for the MSE of 
  the estimate is \cite{demkowicz15quantmetrologysurvey}
\begin{align}
\label{eq:qcrlb}\langle (\theta - \hat{\theta}_n)^2 \rangle &\ge \frac{1}{\mathcal{J}_{\mathrm{Q},n}(\theta)},
\end{align}
where $\mathcal{J}_{\mathrm{Q},n}(\theta)$ is the quantum Fisher information 
  (QFI) associated with the $n$-mode probe state that acquires phase $\theta$ 
  on each mode.
If $n$-mode probe state is a tensor product of $n$ identical probe states, 
  each of which acquires phase $\theta$ independently, then 
\begin{align}
\label{eq:jqn}\mathcal{J}_{\mathrm{Q},n}(\theta)&=n\mathcal{J}_{\mathrm{Q}}(\theta),
\end{align}
  where $\mathcal{J}_{\mathrm{Q}}(\theta)$ is the QFI associated with
  each probe state.

\subsection{Detectability}
The adversary performs a binary hypothesis test on his sample to determine
  whether the target is being interrogated or not.
Performance of the hypothesis test is typically measured by its
  detection error probability
  $\mathbb{P}_{\rm e}^{\rm (det)}=\frac{\mathbb{P}_{\rm FA}+\mathbb{P}_{\rm MD}}{2}$,
  where equal prior probabilities on sensor's interrogation state are
  assumed, $\mathbb{P}_{\rm FA}$ is the probability of false alarm and
  $\mathbb{P}_{\rm MD}$ is the probability of missed detection.
The sensor desires to remain covert by
  ensuring that $\mathbb{P}_{\rm e}^{\rm (det)}\geq\frac{1}{2}-\epsilon$ for
  an arbitrary small $\epsilon>0$ regardless of adversary's measurement
  choice (since $\mathbb{P}_{\rm e}^{\rm (det)}=\frac{1}{2}$ for a random 
  guess).
By decreasing the power used in a probe, the sensor can decrease
  the effectiveness of the adversary's hypothesis test at the expense of
  the increased MSE of the estimate.

\section{Proof of the square root law for covert sensing}\label{sec:proof}
We begin by demonstrating in Section~\ref{sec:converse} that, no matter how one designs the transmitted probe and the measurement (which may include arbitrarily-complicated entangled transmitted states and quantum-limited joint-detection measurements over $n$ modes), the MSE cannot decay any faster than $\mathcal{O}(1/\sqrt{n})$ without the probe being detected by the adversary. Next, we establish the achievability of the SRL for covert phase sensing in Section~\ref{sec:achievability}, where we show that one can attain MSE $\langle(\theta-\hat{\theta}_n)^2\rangle=\mathcal{O}(1/\sqrt{n})$  using laser light illumination and coherent detection.  Finally, we argue for this scheme's near-optimality.

\subsection{Converse}
\label{sec:converse}
Here we show that the SRL for covert phase sensing is 
  insurmountable.
We denote the mean total photon number of the probe sent to the sensing arm 
  using $n$ modes by $\langle N_{\rm S}\rangle=n\bar{n}_{\rm S}$ and 
  the total photon number variance by $\langle\Delta N_{\rm S}^2\rangle$.
Just as in \cite[Theorem 5]{bash15covertbosoniccomm}, we restrict the sensor to
  using $n$-mode probes with total photon number variance
  $\langle \Delta N_{\rm S}^2\rangle=\mathcal{O}(n)$.
However, this restriction is not onerous, as it subsumes all well-known quantum
  states of bosonic mode.

We employ the asymptotic notation \cite[Ch.~3.1]{clrs2e} where 
  $f(n)=\Omega(g(n))$ and $f(n)=\omega(g(n))$ denote asymptotically
  tight and not tight lower bounds on $f(n)$, respectively.

\begin{theorem}[Converse of the square-root law]\label{th:converse}
Suppose the target is interrogated using an $n$-mode probe with a total of
  $\langle N_{\rm S}\rangle=n\bar{n}_{\rm S}$ photons, and that 
  the total photon number variance of the probe is 
  $\langle \Delta N_{\rm S}^2\rangle=\mathcal{O}(n)$.
Then, the sensing attempt is either detected by the adversary with arbitrarily 
  low detection error 
  probability, or the estimator has mean squared error 
  $\langle(\theta-\hat{\theta}_n)^2\rangle=\Omega(1/\sqrt{n})$.
\end{theorem}

\begin{proof}
Suppose the optical interferometer depicted in Figure 
  \ref{fig:setup} uses a general pure state $\ket{\psi}^{P^nR^n}$,
  where $n$ modes are used in both the probe and the reference systems. 
Denoting by $\mathbb{N}_{0}$ the set of all non-negative integers, and by
  $\ket{\mathbf{k}}=\ket{k_1}\otimes\ket{k_2}\otimes\cdots\otimes\ket{k_n}$
  a tensor product of $n$ Fock states, the quantum state of the combined 
  (and potentially entangled) probe and reference states is formally defined as
$\ket{\psi}^{P^nR^n}=\sum_{\mathbf{k}\in\mathbb{N}_{0}^{n}}\sum_{\mathbf{k}^\prime\in\mathbb{N}_{0}^{n}}a_{\mathbf{k}, \mathbf{k}^\prime}\ket{\mathbf{k}}\ket{\mathbf{k}^\prime}$,
  where $\sum_{\mathbf{k}\in\mathbb{N}_{0}^{n}}\sum_{\mathbf{k}^\prime\in\mathbb{N}_{0}^{n}}|a_{\mathbf{k}, \mathbf{k}^\prime}|^2=1$.
The state in each system is obtained by tracing out the other,
  for example, the probe that is used to interrogate the target is
  $\rho^{P^n}=\trace_{R^n}\left(\ket{\psi}^{P^nR^nR^nP^n}\hspace{-4pt}\bra{\psi}\right)$.
Therefore, the mean total photon number in the probe is
  $\langle N_{\rm S}\rangle=\sum_{\mathbf{k}\in\mathbb{N}_{0}^{n}}\sum_{\mathbf{k}^\prime\in\mathbb{N}_{0}^{n}}\left(\sum_{i=1}^nk_i\right)|a_{\mathbf{k}, \mathbf{k}^\prime}|^2$
  and the total photon number variance is
  $\langle \Delta N_{\rm S}^2\rangle=\sum_{\mathbf{k}\in\mathbb{N}_{0}^{n}}\sum_{\mathbf{k}^\prime\in\mathbb{N}_{0}^{n}}\left(\sum_{i=1}^nk_i\right)^2|a_{\mathbf{k}, \mathbf{k}^\prime}|^2-\langle N_{\rm S}\rangle^2=\mathcal{O}(n)$.

Provided that the adversary captures a fraction $\gamma$ of the transmitted 
  photons, where $1-\eta\geq\gamma>0$, in 
  Appendix \ref{app:pe_w_ub} we show 
  that the interrogation attempt is detected with arbitrarily low error probability if
  $\langle N_{\rm S}\rangle=\omega(\sqrt{n})$.
To detect the sensor, the adversary uses a standard threshold test on the total 
  photon count output by a noisy photon 
  number resolving detector.\footnote{We also note that, if the sensor is 
  peak-power constrained (i.e., restricted to a finite photon number per mode),
  then a single photon detector is sufficient.}

When an $n$-mode probe passes through a lossy thermal-noise bosonic channel and
  acquires phase $\theta$ on each mode, we have an upper bound 
  $\mathcal{J}_{\mathrm{Q},n}(\theta)\leq C_{\mathrm{Q},n}(\theta)$, where
  \cite{gagatsos17cqbound}
\begin{align}
\label{eq:C_Q}C_{\mathrm{Q},n}(\theta)&=\frac{4 \eta  \langle N_{\rm S} \rangle \langle \Delta N_{\rm S}^2 \rangle \left(n (1+\bar{n}_{\rm B} (1-\eta)) + \eta \langle N_{\rm S} \rangle \right)}{D},
\end{align}
  with
\begin{align*}
D&=\eta  \langle N_{\rm S} \rangle \left(n (1+\bar{n}_{\rm B} (1-\eta)) + \eta \langle N_{\rm S} \rangle \right)\\
&\quad + (1-\eta)\eta \langle \Delta N_{\rm S}^2 \rangle \langle N_{\rm S} \rangle (1+2 \bar{n}_{\rm B})\\
&\quad - (1-\eta) \eta \langle \Delta N_{\rm S}^2 \rangle n \bar{n}_{\rm B} (1+\bar{n}_{\rm B})\\
&\quad + (1-\eta)n \langle \Delta N_{\rm S}^2 \rangle (1+\bar{n}_{\rm B})^2.
\end{align*}
The sensor must use an $n$-mode probe with 
  $\langle N_{\rm S}\rangle=\mathcal{O}(\sqrt{n})$ 
  photons to avoid detection, which implies that, by
  the QCRLB in \eqref{eq:qcrlb}, the MSE for any estimator of 
  $\theta$ is $\langle (\theta - \hat{\theta}_n)^2 \rangle=\Omega(1/\sqrt{n})$.
\end{proof}

\subsection{Achievability}\label{sec:achievability}

We now prove that the SRL for covert sensing is achievable even when the 
  adversary's capabilities are limited only by the laws of quantum mechanics.
That is, we allow the adversary to collect all the transmitted photons that do 
  not return to  the sensor, perform quantum-limited joint-detection 
  measurements over $n$ modes, and use arbitrary quantum
  computing and storage resources.

\begin{theorem}[Achievability]
\label{th:achievability}
Suppose the sensor attempts to estimate an unknown phase $\theta$ of a pixel 
  using an optical probe that passes through a lossy thermal-noise
  bosonic channel, as described in Figure \ref{fig:setup}.
Also suppose the adversary can perform an arbitrarily complex receiver 
  measurement as permitted by the laws of quantum physics
  and capture all the transmitted photons that do not return to the sensor. 
Then the sensor can lower-bound adversary's detection error probability 
  $\mathbb{P}_{\rm e}^{(\mathrm{det})}\geq\frac{1}{2}-\epsilon$ for any
  $\epsilon>0$ while achieving the MSE
  $\langle(\theta-\hat{\theta}_n)^2\rangle=\mathcal{O}(1/\sqrt{n})$
  using an $n$-mode probe.
\end{theorem}

\begin{proof}
Coherent state is a quantum-mechanical description of ideal laser light.
Let the sensor use an $n$-mode tensor-product coherent state probe 
  $\bigotimes_{i=1}^n\ket{\alpha_i}$ with each $\alpha_i$ drawn 
  independently from an identical zero-mean isotropic complex Gaussian 
  distribution
  $p(\alpha) = e^{-|\alpha|^2/{\bar{n}_{\rm S}}}/{\pi {\bar{n}_{\rm S}}}$, 
  where photon number per state
  $\bar{n}_{\rm S}=\int_{\mathbb C}|\alpha|^2 p(\alpha){\rm d}^2\alpha$. 
Thus, $p(\bigotimes_{i=1}^n\ket{\alpha_i})=\prod_{i=1}^{n}p(\alpha_i)$.
In Appendix \ref{app:pe_w} we show that then
  the probability of detection by the adversary is lower-bounded by
\begin{align}
\label{eq:thermal_pe_lb_ph}\mathbb{P}_{\rm e}^{(\mathrm{det})}&\geq\frac{1}{2}-\frac{(1-\eta)\bar{n}_{\rm S}\sqrt{n}}{4\sqrt{\eta \bar{n}_{\rm B}(1+\eta \bar{n}_{\rm B})}}.
\end{align}
Thus, if the sensor sets
\begin{align}
\label{eq:nbar_ph}\bar{n}_{\rm S}&=\frac{4\epsilon\sqrt{\eta \bar{n}_{\rm B}(1+\eta \bar{n}_{\rm B})}}{\sqrt{n}(1-\eta)},
\end{align}
then he can ensure that the adversary's detection error probability can be 
  lower-bounded by
  $\mathbb{P}_{e}^{(\mathrm{det})}\geq\frac{1}{2}-\epsilon$ over $n$ modes.
In Appendix \ref{app:receiver} we show that the use of an ideal 
  heterodyne receiver achieves the MSE
\begin{align}
\label{eq:het}\langle (\theta - \hat{\theta}_{\mathrm{het},n})^2 \rangle &\approx \frac{c_{\rm het}}{\epsilon \sqrt{n}},
\end{align}
where the constant $c_{\rm het}$ is
\begin{align}
\label{eq:het_const}c_{\rm het}&=\frac{(1-\eta)\left(1+\bar{n}_{\rm B}(1-\eta)\right)}{8\eta\sqrt{\eta \bar{n}_{\rm B}(1+\eta \bar{n}_{\rm B})}}.
\end{align}
Practical heterodyne detectors operate close to the ideal limit, which implies  
 $\langle (\theta - \hat{\theta}_{\mathrm{het},n})^2 \rangle=\mathcal{O}(1/\sqrt{n})$.
\end{proof}

\subsection{The constant in the SRL for covert phase sensing}
\label{sec:limits}
Let's evaluate how far from optimal is the covert phase sensing scheme
  that uses laser light illumination and heterodyne detection, as 
  in the proof of Theorem \ref{th:achievability}.

In Appendix \ref{app:qfi_coh} we show that, when a single-mode 
  coherent state probe is used (with an arbitrary detector), 
  the QFI is
\begin{align}
\label{eq:QFICoh}\mathcal{J}_{\rm Q}^{\rm coh}(\theta)& = \frac{4\bar{n}_{\rm S} \eta}{1 + 2\bar{n}_{\rm B}(1-\eta)}.
\end{align}
Therefore, by \eqref{eq:qcrlb}, \eqref{eq:jqn}, and the substitution of 
  \eqref{eq:nbar_ph} in \eqref{eq:QFICoh}, we have 
  $\langle(\theta - \hat{\theta}_n)^2\rangle\ge\frac{c_{\mathrm{coh}}}{\epsilon \sqrt{n}}$,
  where
\begin{align}
\label{eq:coh_const}c_{\mathrm{coh}} &= \frac{(1-\eta)\left(1+2\bar{n}_{\rm B}(1-\eta)\right)}{16\eta\sqrt{\eta \bar{n}_{\rm B}(1+\eta \bar{n}_{\rm B})}}.
\end{align}
Thus, the MSE attainable using a coherent state probe and an ideal heterodyne 
  receiver is at most twice the quantum limit for a coherent state probe.
We also note that phase can be estimated adaptively using both 
  homodyne and heterodyne receivers \cite{wiseman95adaptive}, potentially 
  closing the gap to \eqref{eq:coh_const}.

Now consider the use of two-mode squeezed vacuum (TMSV) states, where one of 
  the modes is retained as reference while the other is used to probe the phase 
  of the target pixel.
Such states improve the scaling of the MSE in $\bar{n}_{\rm S}$ when there are
  no losses \cite{demkowicz15quantmetrologysurvey}.
The partial trace over one of the modes of the TMSV state yields a thermal
  state with the same Gaussian statistics in the coherent state basis as
  the states used in the proof of Theorem \ref{th:achievability}.
Therefore, since the adversary cannot not access the reference system, we can 
  use the steps in  the proof of Theorem \ref{th:achievability} to show
  the covertness.
In Appendix \ref{app:qfi_sq} we show that the QFI from using 
  the TMSV state is
\begin{align}
\label{eq:QFIsq}\mathcal{J}_{\rm Q}^{\rm sq} = \frac{4 \bar{n}_S (\bar{n}_S+1) \eta}{1+\bar{n}_B (1-\eta)+\bar{n}_S(1-\eta)(1+2 \bar{n}_B)}.
\end{align}
Note that when $\eta=1$ and $\bar{n}_{\rm S}=\mathcal{O}(1)$, 
  $\mathcal{J}_{\rm Q}^{\rm sq} = \mathcal{O}(\bar{n}_{\rm S}^2)$, consistent
  with previous findings that the TMSV states improve the scaling of the MSE in 
  $\bar{n}_{\rm S}$ in lossless scenarios 
  \cite{demkowicz15quantmetrologysurvey}.
However, the substitution of \eqref{eq:nbar_ph} in \eqref{eq:QFIsq}, 
  and the use of \eqref{eq:qcrlb} and \eqref{eq:jqn} yield
  $\langle(\theta - \hat{\theta}_n)^2\rangle\ge\frac{c_{\mathrm{sq}}}{\epsilon \sqrt{n}}$,
  where $c_{\mathrm{sq}}$ is approximated by discarding the low-order terms as
\begin{align}
\label{eq:sq_const}c_{\mathrm{sq}} &\approx \frac{(1-\eta)\left(1+\bar{n}_{\rm B}(1-\eta)\right)}{16\eta\sqrt{\eta \bar{n}_{\rm B}(1+\eta \bar{n}_{\rm B})}}.
\end{align}
The covertness constraint $\bar{n}_{\rm S}=\mathcal{O}(1/\sqrt{n})$ photons/mode
  yields the same scaling of the QFI in $\bar{n}_{\rm S}$ as the 
  coherent state.
While the TMSV state probe outperforms a coherent state probe in phase sensing 
  at high average noise photon number $\bar{n}_{\rm B}$,
  since the constant $c_{\rm het}$ attainable using a coherent state probe and
  heterodyne detection is only twice that of the best attainable constant for
  a TMSV probe $c_{\rm sq}$,
  the challenges associated with using squeezed states may not be worth it.

In fact, we can use the bound \eqref{eq:C_Q} on the QFI for an arbitrary 
  $n$-mode state 
  to derive the ultimate limit for the MSE of phase sensing over the lossy
  thermal-noise bosonic channel.
Since \eqref{eq:C_Q} is increasing in the total photon number variance
  $\langle \Delta N_{\rm S}^2 \rangle$,
  we can upper-bound the QFI as
\begin{align}
\mathcal{J}_{\rm Q}(\theta)&\leq \lim_{\langle \Delta N_{\rm S}^2 \rangle\to\infty} C_{\mathrm{Q},n}(\theta)\nonumber\\
\label{eq:C_Q_max}&=\frac{4 \eta  \langle N_{\rm S} \rangle (n (1+ (1-\eta)\bar{n}_{\rm B})+\eta  \langle N_{\rm S} \rangle)}{(1-\eta)D_{\ell}},
\end{align}
where
\begin{align*}
D_{\ell}&= (1+\bar{n}_{\rm B}) n (1+(1-\eta) \bar{n}_{\rm B})+\eta  (1+2 \bar{n}_{\rm B}) \langle N_{\rm S} \rangle.
\end{align*}
By \eqref{eq:qcrlb}, and the substitution of \eqref{eq:nbar_ph} in 
  \eqref{eq:C_Q_max} (where we note that 
  $\langle N_{\rm S} \rangle=n\bar{n}_{\rm S}$), we have 
  $\langle(\theta - \hat{\theta}_n)^2\rangle\ge\frac{c_{\mathrm{lb}}}{\epsilon \sqrt{n}}$,
  where $c_{\mathrm{lb}}$ is approximated by discarding the low-order terms as
\begin{align}
\label{eq:lb_const}c_{\mathrm{lb}} &\approx \frac{(1-\eta)^2\left(1+\bar{n}_{\rm B}\right)}{16\eta\sqrt{\eta \bar{n}_{\rm B}(1+\eta \bar{n}_{\rm B})}}.
\end{align}
Therefore, the MSE attainable using a practical sensing scheme
  is at most $\frac{2}{1-\eta}$ times the ultimate lower bound.

\section{Discussion}
\label{sec:conclusion}

Section \ref{sec:limits} shows that practically-attainable MSE is a small
  constant factor above optimal.
Moreover, since covertness imposes a strict power constraint
  unlike in other sensing scenarios, here one cannot 
  decrease the MSE by increasing power, 
The only degree of freedom in covert sensing (and communication) is $n$,
  the number of available orthogonal modes.
Now, $n=n_{\rm P}\times n_{\rm S}\times n_{\rm T}$, where $n_{\rm P}=2$ is 
  the number of orthogonal polarizations, $n_{\rm S}$ is the number of 
  orthogonal spatial modes (governed by the channel geometry), and 
  $n_{\rm T}\approx TW$ is the number of temporal modes (time-bandwidth product)
  with $T$ (in seconds) being the transmission time window and $W$ (in Hz) being
  the total spectral bandwidth of the source (see 
  \cite[Supplementary Note 1]{bash15covertbosoniccomm} for a deeper discussion).
Therefore, given a constraint on the available time $T$, one could increase 
  the number of spatial modes $n_{\rm S}$, or increase the spectral bandwidth 
  $W$, or both.  
We will explore this in a follow-up work.

Finally, while here we focus on phase sensing (and assume that the distance
  to the pixel is known), we plan on investigating
  the limits of covert signaling in other sensing tasks such as
  ranging, reflectometry, target detection and classification.
Since the error measures (s.t., the MSE and the probability of error) in many 
  sensing problems are also inversely proportional to the total probe power, 
  we believe that they are governed by the SRLs similar to the one here.
Moreover, simultaneous covert estimation of several parameters (e.g., range and 
  phase) enables covert quantum imaging with its many practical applications.

\begin{widetext}
\appendix

\section{Upper Bound for Adversary's Detection Error Probability in Theorem \ref{th:converse}}
\label{app:pe_w_ub}
Here we adapt the analysis of the adversary's detection error probability from
  the proof of \cite[Theorem 5]{bash15covertbosoniccomm}.

We assume that the sensing arm is lossy and that the adversary has access to 
  the fraction $1-\eta$ leaked photons.
Adversary measures the total photon count $X_{\mathrm{tot}}$ with a noisy 
  photon number resolving (PNR) receiver over the $n$ modes in which 
  the sensor could probe.
For some threshold $S$ (that we discuss later), the adversary declares that 
  the sensor interrogated the target when $X_{\mathrm{tot}}\geq S$, and did not 
  interrogate it when $X_{\mathrm{tot}}<S$.
When the sensor does not interrogate, the adversary observes noise: 
  $X_{\mathrm{tot}}^{(0)}=X_{\rm D}+X_{\rm T}$, where $X_{\rm D}$ 
  is the number of dark counts from the spontaneous emission process at 
  the detector, and $X_{\rm T}$ is the number of photons observed from 
  the thermal background.
We model the dark counts by a Poisson process with rate $\lambda$ 
  photons per mode.
Thus, both the mean and variance of the observed dark counts per
  mode is $\lambda$.
The mean of the number of photons observed per mode from the thermal 
  background with mean photon number per mode $\bar{n}_{\mathrm{B}}$ is 
  $\eta\bar{n}_{\mathrm{B}}$ and the variance is
  $\eta^2(\bar{n}_{\mathrm{B}}+\bar{n}_{\mathrm{B}}^2)$.
Therefore, the mean of the total number of noise photons observed per 
  mode is $\bar{n}_{\rm N}=\lambda+\eta \bar{n}_{\mathrm{B}}$, and, because 
  of the statistical independence of the noise processes, the total variance
  over $n$ modes is
  $\langle \Delta N^2_{\rm N}\rangle=n\lambda+n\eta^2(\bar{n}_{\mathrm{B}}+\bar{n}_{\mathrm{B}}^2)$.
We upper-bound the false alarm probability using Chebyshev's inequality: 
\begin{align}
\mathbb{P}_{\mathrm{FA}}&=\mathbb{P}(X_{\mathrm{tot}}^{(0)}\geq S)\nonumber\\
\label{eq:conv_pfa}&\leq\frac{\langle \Delta N^2_{\rm N}\rangle}{(S-n\bar{n}_{\rm N})^2}.
\end{align}
Thus, to obtain the desired $\mathbb{P}_{\mathrm{FA}}^*$, the adversary sets 
  threshold $S=n\bar{n}_{\rm N}+\sqrt{\langle \Delta N^2_{\rm N}\rangle/\mathbb{P}_{\mathrm{FA}}^*}$.

When the sensor uses a probe $\ket{\psi}^{P^nR^n}$ to interrogate, the adversary
  observes $X_{\mathrm{tot}}^{(1)}=X_u+X_{\rm D}+X_{\rm T}$, 
  where $X_u$ is the count from the transmission of the probe.
We upper-bound the missed detection probability using Chebyshev's inequality:
\begin{align}
\mathbb{P}_{\mathrm{MD}}&=\mathbb{P}(X_{\mathrm{tot}}^{(1)}<S)\nonumber\\
&\leq\mathbb{P}\left(|X_{\mathrm{tot}}^{(1)}-(1-\eta)\langle N_{\rm S}\rangle-\langle \Delta N^2_{\rm N}\rangle|\geq (1-\eta)\langle N_{\rm S}\rangle-\sqrt{\frac{\langle \Delta N^2_{\rm N}\rangle}{\mathbb{P}_{\mathrm{FA}}^*}}\right)\nonumber\\
\label{eq:indep_var}&\leq \frac{\langle \Delta N^2_{\rm N}\rangle+ (1-\eta)^2\langle \Delta N^2_{\rm S}\rangle}{((1-\eta)\langle N_{\rm S}\rangle-\sqrt{\langle \Delta N^2_{\rm N}\rangle/\mathbb{P}_{\mathrm{FA}}^*})^2},
\end{align}
  where equation \eqref{eq:indep_var} is because the noise and the probe
  are independent.
Since $\langle \Delta N^2_{\rm N}\rangle=\mathcal{O}(n)$ and we assume that
  $\langle \Delta N^2_{\rm S}\rangle=\mathcal{O}(n)$, if 
  $\langle N_{\rm S}\rangle=\omega(\sqrt{n})$, then
  $\lim_{n\rightarrow\infty}\mathbb{P}_{\mathrm{MD}}=0$.
Thus, given large enough $n$, the adversary can detect the probes that have
  mean photon number $\langle N_{\rm S}\rangle=\omega(\sqrt{n})$ with 
  probability of error $\mathbb{P}_{\rm e}^{(\mathrm{w})}\leq\epsilon$ for any 
  $\epsilon>0$.

\section{Lower Bound for Adversary's Detection Error Probability}
\label{app:pe_w}
We adapt the analysis of the adversary's detection error probability from
  the proof of \cite[Theorem 2]{bash15covertbosoniccomm}.

When the sensor is not probing the target, the adversary observes thermal 
  environment that is described by the following $n$-copy quantum state 
  (written in Fock state basis):
\begin{align}
\label{eq:rho0}\hat{\rho}_0^{\otimes n}&=\left(\sum_{i=0}^\infty \frac{(\eta \bar{n}_{\rm B})^i}{(1+\eta \bar{n}_{\rm B})^{1+i}}\ket{i}\bra{i}\right)^{\otimes n}
\end{align}

When the sensor probes the target, it first draws a sequence 
  $\boldsymbol{\alpha}=\{\alpha_i\}_{i=1}^n$ of independently and identically 
  distributed (i.i.d.) random variables from a zero-mean isotropic complex 
  Gaussian distribution
  $p(\alpha) = e^{-|\alpha|^2/{\bar{n}_{\rm S}}}/{\pi {\bar{n}_{\rm S}}}$.
It then interrogates the target using $n$-mode tensor-product coherent 
  state probe $\bigotimes_{i=1}^n\ket{\alpha_i}$.
Since the adversary does not have access to $\boldsymbol{\alpha}$, it 
  effectively experiences thermal noise in addition to the environment
  when the sensor probes the target.
Therefore, the following $n$-copy quantum state describes his observation in 
  this case:
\begin{align}
\label{eq:rho1}\hat{\rho}_1^{\otimes n}&=\left(\sum_{i=0}^\infty \frac{((1-\eta)\bar{n}_{\rm S}+\eta \bar{n}_{\rm B})^i}{(1+(1-\eta)\bar{n}_{\rm S}+\eta \bar{n}_{\rm B})^{1+i}}\ket{i}\bra{i}\right)^{\otimes n}.
\end{align}
The adversary has to discriminate between $\hat{\rho}_0$ and $\hat{\rho}_1$
  given in \eqref{eq:rho0} and \eqref{eq:rho1}, respectively.
By \cite[Lemma 2, Supplementary Information]{bash15covertbosoniccomm}, 
  adversary's average probability of discrimination error is:
\begin{align*}
\mathbb{P}_{e}^{(\mathrm{w})}&\geq\frac12\left[1 - \frac12\| \hat{\rho}_1^{\otimes n} - \hat{\rho}_0^{\otimes n} \|_1\right],
\end{align*}
  where a photon number resolving detector achieves the minimum in this case.
The trace distance $\|\hat{\rho}_0-\hat{\rho}_1\|_1$ between states 
  $\hat{\rho}_1$ and $\hat{\rho}_1$ is upper-bounded by the quantum relative 
  entropy (QRE) using quantum Pinsker's 
  Inequality~\cite[Theorem 11.9.5]{wilde13quantumit} as follows:
\begin{align*}
\|\hat{\rho}_0-\hat{\rho}_1\|_1&\leq\sqrt{2D(\hat{\rho}_0\|\hat{\rho}_1)},
\end{align*}
which implies that
\begin{align}
\label{eq:qre_pe_lb}\mathbb{P}_{e}^{(\mathrm{w})}&\geq\frac{1}{2}-\sqrt{\frac{1}{8}D(\hat{\rho}_0^{\otimes n}\|\hat{\rho}_1^{\otimes n})}.
\end{align}
Thus, ensuring that
\begin{align}
\label{eq:tot_qre_bound}D(\hat{\rho}_0^{\otimes n}\|\hat{\rho}_1^{\otimes n})&\leq 8\epsilon^2
\end{align}
ensures that $\mathbb{P}_{e}^{(\mathrm{w})}\geq\frac{1}{2}-\epsilon$ over $n$ modes.
QRE is additive for tensor product states:
\begin{align}
\label{eq:qre_additive}D(\hat{\rho}_0^{\otimes n}\|\hat{\rho}_1^{\otimes n})&=nD(\hat{\rho}_0\|\hat{\rho}_1). 
\end{align}
By \cite[Lemma 4, Supplementary Information]{bash15covertbosoniccomm},
\begin{align}
\label{eq:kl_therm}D(\hat{\rho}_0\|\hat{\rho}_1)&=\eta \bar{n}_{\rm B}\ln\frac{(1+(1-\eta)\bar{n}_{\rm S}+\eta \bar{n}_{\rm B})\eta \bar{n}_{\rm B}}{((1-\eta)\bar{n}_{\rm S}+\eta \bar{n}_{\rm B})(1+\eta \bar{n}_{\rm B})}+\ln\frac{1+(1-\eta)\bar{n}_{\rm S}+\eta \bar{n}_{\rm B}}{1+\eta \bar{n}_{\rm B}}.
\end{align}
The first two terms of the Taylor series expansion of the RHS of
  \eqref{eq:kl_therm} with respect to $\bar{n}_{\rm S}$ at $\bar{n}_{\rm S}=0$ 
  are zero and the fourth term is negative.
Thus, using Taylor's Theorem with the remainder, we can upper-bound 
  equation \eqref{eq:kl_therm} by the third term as follows:
\begin{align}
\label{eq:qre_ub}D(\hat{\rho}_0\|\hat{\rho}_1)&\leq\frac{(1-\eta)^2\bar{n}_{\rm S}^2}{2\eta \bar{n}_{\rm B} (1+\eta \bar{n}_{\rm B})}.
\end{align}
Combining equations \eqref{eq:qre_pe_lb}, \eqref{eq:qre_additive}, and 
  \eqref{eq:qre_ub} yields:
\begin{align}
\label{eq:thermal_pe_lb}\mathbb{P}_{\rm e}^{(\mathrm{w})}&\geq\frac{1}{2}-\frac{(1-\eta)\bar{n}_{\rm S}\sqrt{n}}{4\sqrt{\eta \bar{n}_{\rm B}(1+\eta \bar{n}_{\rm B})}}.
\end{align}

\section{Achievability of the Square Root Law with a Heterodyne Receiver}
\label{app:receiver}
Here we examine the performance of optical heterodyne receiver 
  used with coherent state probes.
We assume ideal shot-noise limited operation
  without a drift in local oscillator (LO) phase.
This assumption is reasonable: we can reduce the impact of excess noise 
  in the receiver by employing a sufficiently powerful LO
  and high-bandwidth electronic components, as well as track the LO phase as
  it drifts.
The sensor satisfies the covertness condition by interrogating 
  the target using $\bar{n}_{\rm S}$ photons/mode, where $\bar{n}_{\rm S}$ is 
  defined in \eqref{eq:nbar_ph}.

When a coherent state acquires a phase shift $\theta$ and is transmitted through
  a lossy-noisy bosonic channel, as depicted in Figure \ref{fig:setup},
  a heterodyne receiver outputs a noisy in-phase and quadrature components of 
  the coherent state that is shifted by $\theta+\phi$, where $\phi$ is 
  the relative phase between the probe and the LO.
We assume that each reading is corrupted by additive white Gaussian noise 
  (AWGN), as is the case in the limit of infinite-power LO 
  \cite{guha04mastersthesis, shapiro08qoc}; in practice, LO power substantially 
  exceeds signal and noise power, ensuring that AWGN is an accurate noise model.
We also assume that the sensor knows the distance to the target
  (simultaneous covert ranging and phase estimation is a challenging
  problem that we plan on addressing in future work).
Thus, the sensor controls $\phi$, and sets it to $\phi=0$.
The noise in the measurement of the in-phase component is independent
  of the noise in the measurement of the quadrature component and vice-versa.

The sensor collects two sequences of observations corresponding to in-phase and
  quadrature components: $\{X^{\rm (I)}_i\}$ and $\{X^{\rm (Q)}_i\}$, 
  $i=1,\ldots,n$.
Here $X^{\rm (I)}_i=\sqrt{\eta\bar{n}_{\rm S}}\cos(\theta)+Z^{\rm (I)}_i$
  and $X^{\rm (Q)}_i=\sqrt{\eta\bar{n}_{\rm S}}\sin(\theta)+Z^{\rm (Q)}_i$,
  with $\{Z^{\rm (I)}_i\}$ and $\{Z^{\rm (Q)}_i\}$
  being sequences of i.i.d. zero-mean Gaussian random variables
  $Z^{\rm (I)}_i\sim\mathcal{N}\left(0,\frac{1+\bar{n}_{\rm B}(1-\eta)}{2}\right)$ and
  $Z^{\rm (Q)}_i\sim\mathcal{N}\left(0,\frac{1+\bar{n}_{\rm B}(1-\eta)}{2}\right)$ 
  \cite{guha04mastersthesis}.
Let's normalize the observations by dividing them by 
  $\sqrt{\eta\bar{n}_{\rm S}}$. 
The resulting sequences are $\{Y^{\rm (I)}_i\}$ and 
  $\{Y^{\rm (Q)}_i\}$, such that
  $Y^{\rm (I)}_i=X^{\rm (I)}_i/\sqrt{\eta\bar{n}_{\rm S}}=\cos(\theta)+Z^{\rm (I,N)}_i$ and 
  $Y^{\rm (Q)}_i=X^{\rm (Q)}_i/\sqrt{\eta\bar{n}_{\rm S}}=\sin(\theta)+Z^{\rm (Q,N)}_i$, where
  $Z^{\rm (I,N)}_i\sim\mathcal{N}\left(0,\frac{1+\bar{n}_{\rm B}(1-\eta)}{2\eta\bar{n}_{\rm S}}\right)$ and
  $Z^{\rm (Q,N)}_i\sim\mathcal{N}\left(0,\frac{1+\bar{n}_{\rm B}(1-\eta)}{2\eta\bar{n}_{\rm S}}\right)$.

Consider the following estimator for $\theta$:
\begin{align}
\label{eq:estimator_het}\hat{\theta}_{\rm het}&=\tan^{-1}\left(\frac{\frac{1}{n}\sum_{i=1}^n Y^{\rm (Q)}_i}{\frac{1}{n}\sum_{i=1}^n Y^{\rm (I)}_i}\right)\\
&=\tan^{-1}\left(\frac{\sin(\theta)+\frac{1}{n}\sum_{i=1}^n Z^{\rm (Q,N)}_i}{\cos(\theta)+\frac{1}{n}\sum_{i=1}^n Z^{\rm (I,N)}_i}\right)\\
&=\tan^{-1}\left(\frac{\sin(\theta)+Z^{\rm (Q)}}{\cos(\theta)+Z^{\rm (I)}}\right),
\end{align}
where $Z^{\rm (I)}\sim\mathcal{N}(0,\sigma^2_{\rm het})$ and
  $Z^{\rm (Q)}\sim\mathcal{N}(0,\sigma^2_{\rm het})$.
The variance $\sigma^2_{\rm het}$ is:
\begin{align}
\label{eq:gauss_add}\sigma^2_{\rm het}&=\frac{1+\bar{n}_{\rm B}(1-\eta)}{2n\eta\bar{n}_{\rm S}}\\
\label{eq:sub_ns}&=\frac{1}{\epsilon \sqrt{n}}\left[\frac{(1-\eta)\left(1+\bar{n}_{\rm B}(1-\eta)\right)}{8\eta\sqrt{\eta \bar{n}_{\rm B}(1+\eta \bar{n}_{\rm B})}}\right],
\end{align}
  where \eqref{eq:gauss_add} is because independent Gaussian random variables
  are additive and \eqref{eq:sub_ns} is from substituting \eqref{eq:nbar_ph}.
The MSE is:
\begin{align}
\left\langle (\theta - {\hat \theta}_{\rm het})^2 \right\rangle &=\left\langle\left(\theta-\tan^{-1}\left(\frac{\sin(\theta)+Z^{\rm (Q)}}{\cos(\theta)+Z^{\rm (I)}}\right)\right)^2 \right\rangle\\
\label{eq:sub_polar}&=\left\langle\left(\theta-\tan^{-1}\left(\frac{\sin(\theta)+R\cos(\varphi)}{\cos(\theta)+R\sin(\varphi)}\right)\right)^2 \right\rangle
\end{align}
where in \eqref{eq:sub_polar} we use circular symmetry of the two-dimensional
  AWGN to change from the rectangular to polar coordinate system.
Thus, the radius is distributed as a Rayleigh random variable
  $R\sim\text{Rayleigh}(\sigma^2_{\rm het})$
  while the angle is distributed uniformly $\varphi\sim\mathcal{U}([0,2\pi])$.
Now, the Taylor series expansion of
  $\tan^{-1}\left(\frac{\sin(\theta)+r\cos(\varphi)}{\cos(\theta)+r\sin(\varphi)}\right)$ around $r=0$ is:
\begin{align}
\tan^{-1}\left(\frac{\sin(\theta)+r\cos(\varphi)}{\cos(\theta)+r\sin(\varphi)}\right)&=\theta+r\cos(\theta+\varphi)-\frac{r^2}{2}\sin(2(\theta+\varphi))-\frac{r^3}{3}\cos(3(\theta+\varphi))+\frac{r^4}{4}\sin(4(\theta+\varphi))\nonumber\\
&\phantom{=}+\frac{r^5}{5}\cos(5(\theta+\varphi))-\frac{r^6}{6}\sin(6(\theta+\varphi))-\frac{r^7}{7}\cos(7(\theta+\varphi))+\frac{r^8}{8}\sin(8(\theta+\varphi))\nonumber\\
\label{eq:taylor}&\phantom{=}+\ldots\\
\label{eq:taylor_ub}&\leq\theta+r\cos(\theta+\varphi)+\sum_{i=2}^\infty \frac{r^i}{i}\\
\label{eq:taylor_log}&=\theta+r\cos(\theta+\varphi)-(\log(1-r)+r)~\text{provided}~0\leq r <1,
\end{align}
where the upper bound in \eqref{eq:taylor_ub} is because 
  $\sin(x),\cos(x)\in[-1,1]$.
While this demonstrates the convergence of the Taylor series converges for 
  $r<1$, the $n^{\mathrm{th}}$ root test shows that the Taylor series in 
  \eqref{eq:taylor} does 
  not converge for $r>1$ (the series converges for $r=1$ by the alternating 
  series test, however, this is a zero-probability event).
However, since $\tan^{-1}(x)\in \left[-\frac{\pi}{2},\frac{\pi}{2}\right]$ and 
  $\theta\in \left(-\frac{\pi}{2},\frac{\pi}{2}\right)$, for any $r$ and 
  $\varphi$, 
\begin{align}
\label{eq:arctan_trivial_bound}\left|\tan^{-1}\left(\frac{\sin(\theta)+r\cos(\varphi)}{\cos(\theta)+r\sin(\varphi)}\right)-\theta\right|\leq \pi.
\end{align}
Therefore, using the Taylor series expansion of $\log(1-x)$ around $x=0$ in
  \eqref{eq:taylor_log}, and \eqref{eq:arctan_trivial_bound}, it is 
  straightforward to show there exist constants $a\in(0,1)$ and $b>0$ such that:
\begin{align}
\label{eq:arctan_ub}\left|\tan^{-1}\left(\frac{\sin(\theta)+r\cos(\varphi)}{\cos(\theta)+r\sin(\varphi)}\right)-\theta\right|\leq \left\{\begin{array}{ll}r\cos(\theta+\varphi)+br^2&\text{if}~r\leq a\\\pi&\text{otherwise}\end{array}\right.
\end{align}
We can use \eqref{eq:arctan_ub} to upper bound the MSE:
\begin{align}
\left\langle (\theta - {\hat \theta}_{\rm het})^2 \right\rangle &\leq\frac{1}{2\pi}\int_0^{2\pi}\left(\int_0^a(r\cos(\theta+\varphi)+br^2)^2\frac{re^{-x^2/2\sigma^2_{\rm het}}}{\sigma^2_{\rm het}}\mathrm{d}r+\int_a^\infty\pi^2\frac{re^{-x^2/2\sigma^2_{\rm het}}}{\sigma^2_{\rm het}}\mathrm{d}r\right)\mathrm{d}\varphi\\
&=\sigma^2_{\rm het}+8b^2\sigma^4_{\rm het}-\frac{1}{2}e^{-\frac{a^2}{2\sigma^2_{\rm het}}}\left(2a^4b^2+a^2(1+8c^2\sigma^2_{\rm het})+2\sigma^2_{\rm het}(1+8c^2\sigma^2_{\rm het})-2\pi^2\right)\\
\label{eq:het_mse_approx}&=\sigma^2_{\rm het}+\mathcal{O}(\sigma^4_{\rm het})\\
&=\frac{1}{\epsilon \sqrt{n}}\left[\frac{(1-\eta)\left(1+\bar{n}_{\rm B}(1-\eta)\right)}{8\eta\sqrt{\eta \bar{n}_{\rm B}(1+\eta \bar{n}_{\rm B})}}\right]+\mathcal{O}\left(\frac{1}{n}\right).
\end{align}

\section{Quantum Fisher information for phase estimation with coherent state and two-mode squeezed vacuum input}
\label{app:qfi}

Here we provide the details of calculating the quantum Fisher information (QFI) for estimating the unknown phase $\theta$ that is picked up by (a) a single-mode coherent state, and (b) by one of the modes of a two-mode squeezed vacuum (TMSV) state.  After the phase shift, the probe state passes through a lossy thermal-noise bosonic channel.  We note that the order of the phase shift and the bosonic channel does not affect the QFI \cite[Appendix A]{tsang13metrology}; we picked the order for the clarity of exposition.

To obtain the QFI we employ the quantum fidelity $F(\hat{\rho}_1,\hat{\rho}_2)$ between two arbitrary Gaussian quantum states $\hat{\rho}_1$ and $\hat{\rho}_2$ \cite{Pirandola2015}:
\begin{align}
F(\hat{\rho}_1,\hat{\rho}_2)&=F_0 \exp \left[ -\frac{1}{4} \bm{\delta}^T (\mathbf{V_1}+\mathbf{V_2})^{-1} \bm{\delta} \right],
\label{appCS:fid}
\end{align}
where $\mathbf{V_1}$ $\left(\mathbf{V_2}\right)$ is the covariance matrix corresponding to the density operator $\hat{\rho}_1$ $\left(\hat{\rho}_2\right)$, $\bm{\delta}$ is the difference of the displacement vectors of each state, 
\begin{align}
\bm{\delta}&=\bm{\delta_2}-\bm{\delta_1},
\label{appCS:delta}
\end{align}
the function $F_0$ is,
\begin{align}
F_0&=\frac{F_{\mathrm{tot}}}{\sqrt[4]{\det{\left(\mathbf{V_1}+\mathbf{V_2}\right)}}}
\label{appCS:F0}
\end{align}
and
\begin{align}
F_{\mathrm{tot}}&= \prod_{k=1}^K \left[w_k+\sqrt{w_k^2-1}\right]^{1/2}.
\label{appCS:FtotProd}
\end{align}
In \eqref{appCS:FtotProd}, $\pm w_k$, $k=1,\ldots,K$ are the (standard) eigenvalues of the matrix
\begin{align}
\mathbf{W}=-2 i \mathbf{V} \bm{\Omega},
\label{appCS:W} 
\end{align}
where $\bm{\Omega}$ is the symplectic invariant matrix
\begin{align}
\bm{\Omega} &=
\begin{pmatrix}
0 & 1 \\
-1 & 0
\end{pmatrix}
\otimes \mathbf{1},
\label{appCS:omega}
\end{align}
$\mathbf{1}$ is the identity matrix, and $\mathbf{V}$ is given by
\begin{align}
\mathbf{V}&= \bm{\Omega}^T \left( \mathbf{V_1}+\mathbf{V_2} \right)^{-1} \left( \frac{\bm{\Omega}}{4} +\mathbf{V_2} \bm{\Omega} \mathbf{V_1} \right).
\label{appCS:V}
\end{align}

We are now equipped to tackle the problem at hand.

\subsection{Coherent State}
\label{app:qfi_coh}
Initially, we have a coherent state $|\alpha \rangle$ with covariance matrix
\begin{align}
\mathbf{V}_{\mathrm{coh}} & = \frac{1}{2}
\begin{pmatrix}
1 & 0 \\
0 & 1
\end{pmatrix}
\end{align}
and displacement vector
\begin{align}
\mathbf{d}_{\mathrm{coh}} & = 
\begin{pmatrix}
\sqrt{2}\Re \alpha \\
\sqrt{2}\Im \alpha
\end{pmatrix}.
\end{align}
The mean photon number of the coherent state is $\bar{n}_{\rm S} = |\alpha|^2 = \Re \alpha^2 +  \Im \alpha^2$.
This coherent state picks up a phase $\theta$, which is described by the phase space transformation
\begin{align}
\mathbf{X}_{\theta} & = 
\begin{pmatrix}
\cos\theta & \sin\theta \\
-\sin\theta & \cos\theta
\end{pmatrix}.
\end{align}
Under this rotation, the covariance matrix of the coherent state remains the same,
\begin{align}
\mathbf{V}_{\mathrm{coh,}\theta} & = \mathbf{X}_{\theta} \mathbf{V}_{\mathrm{coh}} \mathbf{X}_{\theta}^T = \mathbf{V}_{\mathrm{coh}},
\label{appCS:VcohTheta}
\end{align}
while the displacement vector is transformed to:
\begin{align}
\mathbf{d}_{\mathrm{coh,}\theta}= \mathbf{X}_{\theta} \mathbf{d}_{\mathrm{coh}} & = 
\begin{pmatrix}
\sqrt{2}\Re \alpha \cos\theta + \sqrt{2}\Im \alpha \sin \theta\\
\sqrt{2}\Im \alpha \cos\theta - \sqrt{2}\Re \alpha \sin\theta
\end{pmatrix}.
\label{appCS:dcohTheta}
\end{align}
The rotated coherent state now passes through a lossy thermal-noise bosonic channel, where we denote the mean photon number per mode of the thermal environment by $\bar{n}_{\rm B}$. The transformation of the rotated coherent state by the lossy thermal-noise bosonic channel is $\mathbf{V}_{\theta} = \mathbf{X}_{\mathrm{tl}} \mathbf{V}_{\mathrm{coh,}\theta} \mathbf{X}^T_{\mathrm{tl}} + \mathbf{Y}_{\mathrm{tl}}$, $\mathbf{d}_{\theta} = \mathbf{X}_{\mathrm{tl}} \mathbf{d}_{\mathrm{coh,}\theta}+\mathbf{d}_{\mathrm{env}}$, where
\begin{align}
\label{appCS:Xtl}\mathbf{X}_{\mathrm{tl}} & = 
\begin{pmatrix}
\sqrt{\eta} & 0 \\
0 & \sqrt{\eta}
\end{pmatrix}\\
\mathbf{Y}_{\mathrm{tl}} & = (1-\eta)
\begin{pmatrix}
\label{appCS:Ytl}\bar{n}_{\rm B}+\frac{1}{2} & 0 \\
0 & \bar{n}_{\rm B}+\frac{1}{2}
\end{pmatrix}
\end{align}
and the displacement vector of the thermal environment is
\begin{align}
\mathbf{d}_{\mathrm{env}} & = 
\begin{pmatrix}
x_{\mathrm{th}} \\
y_{\mathrm{th}}
\end{pmatrix}.
\label{appCS:denv}
\end{align}
From \eqref{appCS:VcohTheta}, \eqref{appCS:dcohTheta}, \eqref{appCS:Xtl}, \eqref{appCS:Ytl}, and \eqref{appCS:denv} we obtain the final covariance matrix $\mathbf{V_1}$ and the displacement vector $\mathbf{d_1}$
\begin{align}
\label{appCS:V1}\mathbf{V_1} & = 
\begin{pmatrix}
\bar{n}_{\rm B} (1-\eta) +\frac{1}{2} & 0 \\
0 & \bar{n}_{\rm B} (1-\eta) +\frac{1}{2}
\end{pmatrix}\\
\label{appCS:d1}\mathbf{d_1} & = 
\begin{pmatrix}
x_\mathrm{th} + \sqrt{\eta} \sqrt{2} \left( \Re\alpha \cos\theta + \Im\alpha \sin\theta \right)  \\
y_\mathrm{th} + \sqrt{\eta} \sqrt{2} \left( \Im\alpha \cos\theta - \Re\alpha \sin\theta \right)
\end{pmatrix}.
\end{align}
We want to compute the quantum fidelity between the final state described by $\mathbf{V_1}$ and $\mathbf{d_1}$ and a state evolved by $d\theta$ in parameter space, i.e., a state with covariance matrix
\begin{align}
\mathbf{V_2} & \equiv \mathbf{V_1}({\theta \rightarrow \theta+d\theta})=\mathbf{V_1}
\label{appCS:V2}
\end{align}
and displacement vector
\begin{align}
\mathbf{d_2} & \equiv \mathbf{d_1}({\theta \rightarrow \theta+d\theta})=
\begin{pmatrix}
x_\mathrm{th} + \sqrt{\eta} \sqrt{2} \left( \Re\alpha \cos(\theta+d\theta) + \Im\alpha \sin(\theta+d\theta) \right)  \\
y_\mathrm{th} + \sqrt{\eta} \sqrt{2} \left( \Im\alpha \cos(\theta+d\theta) - \Re\alpha \sin(\theta+d\theta) \right)
\end{pmatrix}.
\label{appCS:d2}
\end{align}
Using \eqref{appCS:W}, \eqref{appCS:omega}, \eqref{appCS:V}, \eqref{appCS:V1}, and \eqref{appCS:V2} we derive:
\begin{align}
\mathbf{W} & = \frac{2 i}{1 + \bar{n}_{\rm B} (1-\eta)}
\begin{pmatrix}
0 & - \left( \bar{n}_{\rm B} (1-\eta) + \frac{1}{2} \right)^2 -\frac{1}{4} \\
\left( \bar{n}_{\rm B} (1-\eta) + \frac{1}{2} \right)^2 + \frac{1}{4} & 0
\end{pmatrix}.
\label{appCS:Wfinal}
\end{align} 
The form of $\mathbf{W}$ as expressed in \eqref{appCS:Wfinal} implies that it has two eigenvalues $\pm w_1$ with,
\begin{align}
w_1 & =\frac{2}{1 + \bar{n}_{\rm B} (1-\eta)}\left( \left( \bar{n}_{\rm B} (1-\eta) + \frac{1}{2} \right)^2 + \frac{1}{4} \right).
\label{appCS:eig}
\end{align}  
Using \eqref{appCS:delta}, \eqref{appCS:F0}, \eqref{appCS:FtotProd}, \eqref{appCS:V1}, \eqref{appCS:d1}, \eqref{appCS:V2}, and \eqref{appCS:d2} we find the quantum fidelity $F(\hat{\rho}_1,\hat{\rho}_2) \equiv F(d\theta)$ in (\ref{appCS:fid}). The QFI is given by four times the second order term of the expansion of $1-F(d\theta)$ (the $1/2$ factor in front of the expansion's second order term is not included):
\begin{align} 
\mathcal{J}_{\rm Q} & = 4 \frac{d^2}{d(d\theta)^2} (1-F(d\theta))\Bigg|_{d\theta=0} = \frac{4 \bar{n}_{\rm S} \eta}{1+2 \bar{n}_{\rm B}(1-\eta)}.
\end{align} 

\subsection{Two-mode Squeezed Vacuum (TMSV) State}
\label{app:qfi_sq}

The covariance matrix of the TMSV state is:
\begin{align}
\mathbf{V}_{\textrm{sq}} &= \frac{1}{2}
\begin{pmatrix}
\cosh 2|\xi|  & \sinh 2|\xi| & 0             & 0 \\
\sinh 2 |\xi| & \cosh 2|\xi| & 0             & 0  \\
  0           &  0           & \cosh 2|\xi|  & -\sinh 2|\xi|  \\
  0           &  0           & -\sinh 2|\xi| &\cosh 2|\xi|
\end{pmatrix},
\label{Vsq}
\end{align}
noting that the coordinates representation we use is of the form $\left(q_1,q_2, \ldots, p_1, p_2, \ldots \right)$. Also, the TMSV state is expressed in Fock (photon number) basis as follows:
\begin{eqnarray}
|00;|\xi| \rangle = \frac{1}{\cosh|\xi|} \sum_{k} \tanh|\xi|^k |kk\rangle.
\end{eqnarray}
We use one of the modes of TMSV state for probing the unknown, and keep the other as reference.
The mean photon number for either mode of the TMSV state is $\bar{n}_{\rm S} = \sinh^2 |\xi|$.

The probing mode of the TMSV passes through the lossy thermal-noise bosonic channel, while nothing happens to the reference mode. Therefore, the symplectic phase transformation is as follows:
\begin{align}
\mathbf{X'}_{\theta} &=
\begin{pmatrix}
 \cos \theta & 0 & \sin \theta & 0 \\
 0 & 1  & 0 &  0\\
 -\sin\theta & 0 & \cos\theta & 0\\
 0 & 0 & 0 & 1
\end{pmatrix},
\end{align} 
 and the complete channel transformation, i.e., lossy thermal-noise bosonic channel for the probing mode and identity for the reference modes, is as follows:
\begin{align}
\mathbf{X'}_{\rm tl} &= 
\begin{pmatrix}
\sqrt{\eta} & 0 & 0 & 0 \\
 0 & 1 & 0 & 0\\
 0 & 0 & \sqrt{\eta} & 0 \\
 0 & 0 & 0 & 1
\end{pmatrix}\\
\mathbf{Y'}_{\rm tl} &= (1-\eta)
\begin{pmatrix}
\bar{n}_{\rm B}+\frac{1}{2} & 0 & 0 & 0 \\
 0 & 0 & 0 & 0\\
 0 & 0 & \bar{n}_{\rm B}+\frac{1}{2} & 0 \\
 0 & 0 & 0 & 0
\end{pmatrix},
\end{align}   
where $\eta$ is the transmittance of the channel and $\bar{n}_{\rm B}$ is the thermal background mean photon number.

The output covariance matrix is $ \mathbf{V_1} = \mathbf{X'}_{\rm tl} \mathbf{X'}_{\theta} \mathbf{V}_{\rm sq} \mathbf{X'}_{\theta}^T \mathbf{X'}_{\rm tl} + \mathbf{Y'}_{\rm tl} $. Note that the displacement vector of the TMSV state is a zero-vector and the phase shifting information is carried by the output covariance matrix. Following the same procedure as in (\ref{appCS:V2}) and (\ref{appCS:Wfinal}) for the coherent state state probe, and the computing the eigenvalues of the latter, we find the QFI for estimating phase $\theta$ using a TMSV state:

\begin{align} 
\mathcal{J}_{\rm Q}^{\rm sq} = \frac{2\eta \sinh^2 2|\xi|}{1+\eta+(1+2 \bar{n}_{\rm B})(1-\eta) \cosh 2|\xi|}.
\label{AppCS:QFIxi}
\end{align}
Noting that $\bar{n}_{\rm S} = \sinh^2 |\xi| \rightarrow |\xi| = \arcsinh \sqrt{\bar{n}_{\rm S}}$, (\ref{AppCS:QFIxi}) can be written as:
\begin{align} 
\mathcal{J}_{\rm Q}^{\rm sq} = \frac{4 \bar{n}_{\rm S} (\bar{n}_{\rm S}+1) \eta}{1+\bar{n}_{\rm B} (1-\eta)+\bar{n}_{\rm S}(1-\eta)(1+2 \bar{n}_{\rm B})}.
\label{AppCS:QFIns}
\end{align}

\end{widetext}

\begin{thebibliography}{16}%
\makeatletter
\providecommand \@ifxundefined [1]{%
 \@ifx{#1\undefined}
}%
\providecommand \@ifnum [1]{%
 \ifnum #1\expandafter \@firstoftwo
 \else \expandafter \@secondoftwo
 \fi
}%
\providecommand \@ifx [1]{%
 \ifx #1\expandafter \@firstoftwo
 \else \expandafter \@secondoftwo
 \fi
}%
\providecommand \natexlab [1]{#1}%
\providecommand \enquote  [1]{``#1''}%
\providecommand \bibnamefont  [1]{#1}%
\providecommand \bibfnamefont [1]{#1}%
\providecommand \citenamefont [1]{#1}%
\providecommand \href@noop [0]{\@secondoftwo}%
\providecommand \href [0]{\begingroup \@sanitize@url \@href}%
\providecommand \@href[1]{\@@startlink{#1}\@@href}%
\providecommand \@@href[1]{\endgroup#1\@@endlink}%
\providecommand \@sanitize@url [0]{\catcode `\\12\catcode `\$12\catcode
  `\&12\catcode `\#12\catcode `\^12\catcode `\_12\catcode `\%12\relax}%
\providecommand \@@startlink[1]{}%
\providecommand \@@endlink[0]{}%
\providecommand \url  [0]{\begingroup\@sanitize@url \@url }%
\providecommand \@url [1]{\endgroup\@href {#1}{\urlprefix }}%
\providecommand \urlprefix  [0]{URL }%
\providecommand \Eprint [0]{\href }%
\providecommand \doibase [0]{http://dx.doi.org/}%
\providecommand \selectlanguage [0]{\@gobble}%
\providecommand \bibinfo  [0]{\@secondoftwo}%
\providecommand \bibfield  [0]{\@secondoftwo}%
\providecommand \translation [1]{[#1]}%
\providecommand \BibitemOpen [0]{}%
\providecommand \bibitemStop [0]{}%
\providecommand \bibitemNoStop [0]{.\EOS\space}%
\providecommand \EOS [0]{\spacefactor3000\relax}%
\providecommand \BibitemShut  [1]{\csname bibitem#1\endcsname}%
\let\auto@bib@innerbib\@empty
\bibitem [{\citenamefont {Bash}\ \emph {et~al.}(2013)\citenamefont {Bash},
  \citenamefont {Goeckel},\ and\ \citenamefont
  {Towsley}}]{bash13squarerootjsacisit}%
  \BibitemOpen
  \bibfield  {author} {\bibinfo {author} {\bibfnamefont {Boulat~A.}\
  \bibnamefont {Bash}}, \bibinfo {author} {\bibfnamefont {Dennis}\ \bibnamefont
  {Goeckel}}, \ and\ \bibinfo {author} {\bibfnamefont {Don}\ \bibnamefont
  {Towsley}},\ }\bibfield  {title} {\enquote {\bibinfo {title} {Limits of
  reliable communication with low probability of detection on {AWGN}
  channels},}\ }\href@noop {} {\bibfield  {journal} {\bibinfo  {journal}
  {{IEEE} J. Select. Areas Commun.}\ }\textbf {\bibinfo {volume} {31}},\
  \bibinfo {pages} {1921--1930} (\bibinfo {year} {2013})},\ \bibinfo {note}
  {{Originally presented at ISIT 2012, Cambridge MA}},\ \Eprint
  {http://arxiv.org/abs/arXiv:1202.6423} {arXiv:1202.6423} \BibitemShut
  {NoStop}%
\bibitem [{\citenamefont {Che}\ \emph {et~al.}(2013)\citenamefont {Che},
  \citenamefont {Bakshi},\ and\ \citenamefont {Jaggi}}]{che13sqrtlawbscisit}%
  \BibitemOpen
  \bibfield  {author} {\bibinfo {author} {\bibfnamefont {Pak~Hou}\ \bibnamefont
  {Che}}, \bibinfo {author} {\bibfnamefont {Mayank}\ \bibnamefont {Bakshi}}, \
  and\ \bibinfo {author} {\bibfnamefont {Sidharth}\ \bibnamefont {Jaggi}},\
  }\bibfield  {title} {\enquote {\bibinfo {title} {Reliable deniable
  communication: Hiding messages in noise},}\ }in\ \href@noop {} {\emph
  {\bibinfo {booktitle} {Proc. {IEEE} Int. Symp. Inform. Theory ({ISIT})}}}\
  (\bibinfo {address} {Istanbul, Turkey},\ \bibinfo {year} {2013})\ \bibinfo
  {note} {arXiv:1304.6693}\BibitemShut {NoStop}%
\bibitem [{\citenamefont {Bloch}(2016)}]{bloch15covert}%
  \BibitemOpen
  \bibfield  {author} {\bibinfo {author} {\bibfnamefont {M.~R.}\ \bibnamefont
  {Bloch}},\ }\bibfield  {title} {\enquote {\bibinfo {title} {Covert
  communication over noisy channels: A resolvability perspective},}\ }\href
  {\doibase 10.1109/TIT.2016.2530089} {\bibfield  {journal} {\bibinfo
  {journal} {IEEE Trans. Inf. Theory}\ }\textbf {\bibinfo {volume} {62}},\
  \bibinfo {pages} {2334--2354} (\bibinfo {year} {2016})}\BibitemShut {NoStop}%
\bibitem [{\citenamefont {Wang}\ \emph {et~al.}(2016)\citenamefont {Wang},
  \citenamefont {Wornell},\ and\ \citenamefont {Zheng}}]{wang15covert}%
  \BibitemOpen
  \bibfield  {author} {\bibinfo {author} {\bibfnamefont {L.}~\bibnamefont
  {Wang}}, \bibinfo {author} {\bibfnamefont {G.~W.}\ \bibnamefont {Wornell}}, \
  and\ \bibinfo {author} {\bibfnamefont {L.}~\bibnamefont {Zheng}},\ }\bibfield
   {title} {\enquote {\bibinfo {title} {Fundamental limits of communication
  with low probability of detection},}\ }\href {\doibase
  10.1109/TIT.2016.2548471} {\bibfield  {journal} {\bibinfo  {journal} {IEEE
  Trans. Inf. Theory}\ }\textbf {\bibinfo {volume} {62}},\ \bibinfo {pages}
  {3493--3503} (\bibinfo {year} {2016})}\BibitemShut {NoStop}%
\bibitem [{\citenamefont {Bash}\ \emph
  {et~al.}(2015{\natexlab{a}})\citenamefont {Bash}, \citenamefont {Gheorghe},
  \citenamefont {Patel}, \citenamefont {Habif}, \citenamefont {Goeckel},
  \citenamefont {Towsley},\ and\ \citenamefont
  {Guha}}]{bash15covertbosoniccomm}%
  \BibitemOpen
  \bibfield  {author} {\bibinfo {author} {\bibfnamefont {Boulat~A.}\
  \bibnamefont {Bash}}, \bibinfo {author} {\bibfnamefont {Andrei~H.}\
  \bibnamefont {Gheorghe}}, \bibinfo {author} {\bibfnamefont {Monika}\
  \bibnamefont {Patel}}, \bibinfo {author} {\bibfnamefont {Jonathan~L.}\
  \bibnamefont {Habif}}, \bibinfo {author} {\bibfnamefont {Dennis}\
  \bibnamefont {Goeckel}}, \bibinfo {author} {\bibfnamefont {Don}\ \bibnamefont
  {Towsley}}, \ and\ \bibinfo {author} {\bibfnamefont {Saikat}\ \bibnamefont
  {Guha}},\ }\bibfield  {title} {\enquote {\bibinfo {title} {Quantum-secure
  covert communication on bosonic channels},}\ }\href {\doibase
  10.1038/NCOMMS9626} {\bibfield  {journal} {\bibinfo  {journal} {Nat Commun}\
  }\textbf {\bibinfo {volume} {6}} (\bibinfo {year} {2015}{\natexlab{a}}),\
  10.1038/NCOMMS9626}\BibitemShut {NoStop}%
\bibitem [{\citenamefont {Bash}\ \emph
  {et~al.}(2015{\natexlab{b}})\citenamefont {Bash}, \citenamefont {Goeckel},
  \citenamefont {Guha},\ and\ \citenamefont {Towsley}}]{bash15covertcommmag}%
  \BibitemOpen
  \bibfield  {author} {\bibinfo {author} {\bibfnamefont {Boulat~A.}\
  \bibnamefont {Bash}}, \bibinfo {author} {\bibfnamefont {Dennis}\ \bibnamefont
  {Goeckel}}, \bibinfo {author} {\bibfnamefont {Saikat}\ \bibnamefont {Guha}},
  \ and\ \bibinfo {author} {\bibfnamefont {Don}\ \bibnamefont {Towsley}},\
  }\bibfield  {title} {\enquote {\bibinfo {title} {Hiding information in noise:
  Fundamental limits of covert wireless communication},}\ }\href@noop {}
  {\bibfield  {journal} {\bibinfo  {journal} {{IEEE} Commun. Mag.}\ }\textbf
  {\bibinfo {volume} {53}} (\bibinfo {year} {2015}{\natexlab{b}})},\ \Eprint
  {http://arxiv.org/abs/arXiv:1506.00066} {arXiv:1506.00066} \BibitemShut
  {NoStop}%
\bibitem [{\citenamefont {Tsang}(2013)}]{tsang13metrology}%
  \BibitemOpen
  \bibfield  {author} {\bibinfo {author} {\bibfnamefont {Mankei}\ \bibnamefont
  {Tsang}},\ }\bibfield  {title} {\enquote {\bibinfo {title} {Quantum metrology
  with open dynamical systems},}\ }\href {\doibase
  10.1088/1367-2630/15/7/073005} {\bibfield  {journal} {\bibinfo  {journal}
  {New Journal of Physics}\ }\textbf {\bibinfo {volume} {15}},\ \bibinfo
  {pages} {073005} (\bibinfo {year} {2013})}\BibitemShut {NoStop}%
\bibitem [{\citenamefont {Malik}\ \emph {et~al.}(2012)\citenamefont {Malik},
  \citenamefont {Maga{\~n}a-Loaiza},\ and\ \citenamefont
  {Boyd}}]{malik12quantumsecuredimaging}%
  \BibitemOpen
  \bibfield  {author} {\bibinfo {author} {\bibfnamefont {Mehul}\ \bibnamefont
  {Malik}}, \bibinfo {author} {\bibfnamefont {Omar~S.}\ \bibnamefont
  {Maga{\~n}a-Loaiza}}, \ and\ \bibinfo {author} {\bibfnamefont {Robert~W.}\
  \bibnamefont {Boyd}},\ }\bibfield  {title} {\enquote {\bibinfo {title}
  {Quantum-secured imaging},}\ }\href {\doibase 10.1063/1.4770298} {\bibfield
  {journal} {\bibinfo  {journal} {Applied Physics Letters}\ }\textbf {\bibinfo
  {volume} {101}},\ \bibinfo {pages} {241103} (\bibinfo {year} {2012})},\
  \Eprint {http://arxiv.org/abs/arXiv:1212.2605 [quant-ph]} {arXiv:1212.2605
  [quant-ph]} \BibitemShut {NoStop}%
\bibitem [{\citenamefont {Demkowicz-Dobrza\'{n}ski}\ \emph
  {et~al.}(2015)\citenamefont {Demkowicz-Dobrza\'{n}ski}, \citenamefont
  {Jarzyna},\ and\ \citenamefont
  {Ko\l{}ody\'{n}ski}}]{demkowicz15quantmetrologysurvey}%
  \BibitemOpen
  \bibfield  {author} {\bibinfo {author} {\bibfnamefont {R.}~\bibnamefont
  {Demkowicz-Dobrza\'{n}ski}}, \bibinfo {author} {\bibfnamefont
  {M.}~\bibnamefont {Jarzyna}}, \ and\ \bibinfo {author} {\bibfnamefont
  {J.}~\bibnamefont {Ko\l{}ody\'{n}ski}},\ }\bibfield  {title} {\enquote
  {\bibinfo {title} {Quantum limits in optical interferometry},}\ }\href
  {\doibase 10.1016/bs.po.2015.02.003} {\bibfield  {journal} {\bibinfo
  {journal} {Progress in Optics}\ }\textbf {\bibinfo {volume} {60}},\ \bibinfo
  {pages} {345} (\bibinfo {year} {2015})},\ \Eprint
  {http://arxiv.org/abs/arXiv:1405.7703 [quant-ph]} {arXiv:1405.7703
  [quant-ph]} \BibitemShut {NoStop}%
\bibitem [{\citenamefont {Cormen}\ \emph {et~al.}(2001)\citenamefont {Cormen},
  \citenamefont {Leiserson}, \citenamefont {Rivest},\ and\ \citenamefont
  {Stein}}]{clrs2e}%
  \BibitemOpen
  \bibfield  {author} {\bibinfo {author} {\bibfnamefont {Thomas~H.}\
  \bibnamefont {Cormen}}, \bibinfo {author} {\bibfnamefont {Charles~E.}\
  \bibnamefont {Leiserson}}, \bibinfo {author} {\bibfnamefont {Ronald~L.}\
  \bibnamefont {Rivest}}, \ and\ \bibinfo {author} {\bibfnamefont {Clifford}\
  \bibnamefont {Stein}},\ }\href@noop {} {\emph {\bibinfo {title} {Introduction
  to Algorithms}}},\ \bibinfo {edition} {2nd}\ ed.\ (\bibinfo  {publisher}
  {{MIT} Press},\ \bibinfo {address} {Cambridge, Massachusetts},\ \bibinfo
  {year} {2001})\BibitemShut {NoStop}%
\bibitem [{\citenamefont {Gagatsos}\ \emph {et~al.}(2017)\citenamefont
  {Gagatsos}, \citenamefont {Bash}, \citenamefont {Guha},\ and\ \citenamefont
  {Datta}}]{gagatsos17cqbound}%
  \BibitemOpen
  \bibfield  {author} {\bibinfo {author} {\bibfnamefont {Christos~N.}\
  \bibnamefont {Gagatsos}}, \bibinfo {author} {\bibfnamefont {Boulat~A.}\
  \bibnamefont {Bash}}, \bibinfo {author} {\bibfnamefont {Saikat}\ \bibnamefont
  {Guha}}, \ and\ \bibinfo {author} {\bibfnamefont {Animesh}\ \bibnamefont
  {Datta}},\ }\href@noop {} {\enquote {\bibinfo {title} {On bounding the
  quantum limits of estimation through a thermal loss channel},}\ }\bibinfo
  {howpublished} {arXiv:1701.05518 [quant-ph]} (\bibinfo {year}
  {2017})\BibitemShut {NoStop}%
\bibitem [{\citenamefont {Wiseman}(1995)}]{wiseman95adaptive}%
  \BibitemOpen
  \bibfield  {author} {\bibinfo {author} {\bibfnamefont {H.~M.}\ \bibnamefont
  {Wiseman}},\ }\bibfield  {title} {\enquote {\bibinfo {title} {Adaptive phase
  measurements of optical modes: Going beyond the marginal $q$ distribution},}\
  }\href {\doibase 10.1103/PhysRevLett.75.4587} {\bibfield  {journal} {\bibinfo
   {journal} {Phys. Rev. Lett.}\ }\textbf {\bibinfo {volume} {75}},\ \bibinfo
  {pages} {4587--4590} (\bibinfo {year} {1995})}\BibitemShut {NoStop}%
\bibitem [{\citenamefont {Wilde}(2013)}]{wilde13quantumit}%
  \BibitemOpen
  \bibfield  {author} {\bibinfo {author} {\bibfnamefont {M.M.}\ \bibnamefont
  {Wilde}},\ }\href@noop {} {\emph {\bibinfo {title} {Quantum Information
  Theory}}}\ (\bibinfo  {publisher} {Cambridge University Press},\ \bibinfo
  {year} {2013})\BibitemShut {NoStop}%
\bibitem [{\citenamefont {Guha}(2004)}]{guha04mastersthesis}%
  \BibitemOpen
  \bibfield  {author} {\bibinfo {author} {\bibfnamefont {Saikat}\ \bibnamefont
  {Guha}},\ }\emph {\bibinfo {title} {Classical Capacity of the Free-Space
  Quantum-Optical Channel}},\ \href@noop {} {Master's thesis},\ \bibinfo
  {school} {Massachusetts Institute of Technology} (\bibinfo {year}
  {2004})\BibitemShut {NoStop}%
\bibitem [{\citenamefont {Shapiro}(Spring 2008)}]{shapiro08qoc}%
  \BibitemOpen
  \bibfield  {author} {\bibinfo {author} {\bibfnamefont {J.~H.}\ \bibnamefont
  {Shapiro}},\ }\href@noop {} {\enquote {\bibinfo {title} {{6.453 Quantum
  Optical Communication}},}\ }\bibinfo {howpublished} {Massachusetts Institute
  of Technology: MIT OpenCouseWare} (\bibinfo {year} {Spring 2008}),\ \bibinfo
  {note} {\url{http://ocw.mit.edu } (Accessed June 15, 2016)}\BibitemShut
  {NoStop}%
\bibitem [{\citenamefont {Banchi}\ \emph {et~al.}(2015)\citenamefont {Banchi},
  \citenamefont {Braunstein},\ and\ \citenamefont {Pirandola}}]{Pirandola2015}%
  \BibitemOpen
  \bibfield  {author} {\bibinfo {author} {\bibfnamefont {Leonardo}\
  \bibnamefont {Banchi}}, \bibinfo {author} {\bibfnamefont {Samuel~L.}\
  \bibnamefont {Braunstein}}, \ and\ \bibinfo {author} {\bibfnamefont
  {Stefano}\ \bibnamefont {Pirandola}},\ }\bibfield  {title} {\enquote
  {\bibinfo {title} {Quantum fidelity for arbitrary gaussian states},}\ }\href
  {\doibase 10.1103/PhysRevLett.115.260501} {\bibfield  {journal} {\bibinfo
  {journal} {Phys. Rev. Lett.}\ }\textbf {\bibinfo {volume} {115}},\ \bibinfo
  {pages} {260501} (\bibinfo {year} {2015})}\BibitemShut {NoStop}%
\end{thebibliography}
\end{document}